\documentclass[journal,draftcls,onecolumn,12pt,twoside]{IEEEtran}

\usepackage[english]{babel}
\usepackage{amsmath}
\usepackage{amssymb}
\usepackage{amsthm}
\usepackage{amstext}
\usepackage{graphicx}
\usepackage{float}
\usepackage{booktabs}
\usepackage{microtype}
\usepackage{algorithm}
\usepackage[noend]{algorithmic}
\usepackage{url}
\usepackage{cite}
\usepackage{color}
\usepackage{xcolor}
\usepackage{bm}
\UseRawInputEncoding
\addtocounter{MaxMatrixCols}{10}

\newcommand{\eq}{\mathsf{\,=\,}}

\newtheorem{theorem}{Theorem}

\newtheorem{lemma}[theorem]{Lemma}
\newtheorem{corollary}[theorem]{Corollary}

\newtheorem{example}{Example}

\long\def\symbolfootnote[#1]#2{\begingroup
\def\thefootnote{\fnsymbol{footnote}}\footnote[#1]{#2}\endgroup}
\IEEEoverridecommandlockouts

\title{On MDS Condition and Erased Lines Recovery of Generalized Expanded-Blaum-Roth Codes
and Generalized Blaum-Roth Codes}

\author{Hanxu Hou,~\IEEEmembership{Member,~IEEE}, and Mario
Blaum,~\IEEEmembership{Life Fellow,~IEEE}
\thanks{
H. Hou is with the School of Electrical Engineering \&
Intelligentization, Dongguan University of Technology~(E-mail:
houhanxu@163.com).
M. Blaum is with Group of Analysis, Security and Systems (GASS),
Faculty of Computer Science and Engineering,
Universidad Complutense de Madrid (UCM),
Ciudad Universitaria, 28040 Madrid, Spain (E-mail: mblaum@ieee.org).
This work was partially supported by the National Key R\&D Program of
China (No. 2020YFA0712300), the National Natural Science Foundation of China (No. 62071121),
the Key Area Research and Development Program of Guangdong Province (2020B0101110003) and Basic Research Enhancement Program of China under Grant 2021-JCJQ-JJ-0483.}}

\begin{document}

%\bibliographystyle{plain}

%\markboth{IEEE Transactions on Information Theory}%
%{Submitted paper}

\maketitle

\begin{abstract}
Generalized Expanded-Blaum-Roth (GEBR) codes \cite{wu2020} are designed for large-scale distributed
storage systems that have larger recoverability for single-symbol failures, multi-column failures
and multi-row failures, compared with locally recoverable codes (LRC).
GEBR codes encode an $\alpha\times k$
information array into a $p\tau\times (k+r)$ array such that lines of slope
$i$ with $0\leq i\leq r-1$ have even parity and each column contains $p\tau-\alpha$ local
parity symbols, where $p$ is an odd prime and $k+r\leq p\tau$.
Necessary and sufficient conditions for GEBR codes to be $(n,k)$ recoverable
(i.e., any $k$ out of $n=k+r$ columns can retrieve all information symbols) are given in \cite{hou2021generalization}
for $\alpha=(p-1)\tau$. However, the $(n,k)$ recoverable condition of GEBR codes is unknown
when $\alpha<(p-1)\tau$. In this paper, we present the $(n,k)$ recoverable condition for GEBR codes for
$\alpha< (p-1)\tau$.
In addition, we present a sufficient condition for enabling GEBR codes
to recover some erased lines of any slope $i$ ($0\leq i\leq p\tau-1$) for
any parameter $r$ when $\tau$ is a power of $p$.
Moreover, we present the construction of Generalized Blaum-Roth (GBR) codes that
encode an $\alpha\times k$ information array into an $\alpha\times (k+r)$ array.
We show that GBR codes share the same MDS condition as the $(n,k)$ recoverable condition of GEBR
codes, and we
also present a sufficient condition for GBR codes to recover some erased lines
of any slope $i$ ($0\leq i\leq \alpha-1$).
\end{abstract}


\section{Introduction}

Modern distributed storage systems usually employ erasure coding to maintain data
availability and durability in the presence of failures that can deliver high
data reliability with a small storage overhead.
Array codes are a class of erasure-correcting codes with only XOR and
cyclic-shift operations being involved in the coding process. They
have been widely used in storage systems, such as Redundant Arrays of
Independent Disks (RAID) \cite{Patterson1989}. An $m\times (k+r)$ array code
is encoded from $m\times k$ information symbols. % that is stored in $k+r$
%columns each with $m$ symbols.
{\em Maximum distance separable (MDS)} array codes \cite{Blaum1995,Corbett2004Ro,4358290,hou2014,hou2021star}
are a special class of
array codes, where any $k$ out of the $k+r$ columns can retrieve all
$m\times k$ information symbols.
	
In large-scale distributed storage systems, the data files
are often geographically distributed across nodes, racks and data
centers.
%It is necessary to access
The data files may need to be accessed even if some nodes, racks or
data centers are temporarily off-line.
Therefore, it it important to design array codes optimizing their rate for a given
storage overhead in large-scale storage systems %to optimize the storage efficient
such that both single-symbol failures and correlated multi-symbol failures can be
efficiently recovered.
Locally recoverable codes (LRC) \cite{2014A} %are such array codes that
have local repair for one or more symbol failures by adding parity
symbols for (generally disjoint) groups of
symbols. However, although optimal LRC codes maximize the minimum
distance of the codes, they are not designed to deal with erasures of
whole groups (which in general are columns in the aforementioned
arrays).
%if more symbols are erased in a group, LRC codes cannot locally recover
%the erased symbols. Moreover, there is no fast decoding algorithm designed for LRC
%when the correlated multi-symbols are erased.
Recently, $p\times (k+r)$ Expanded Blaum-Roth (EBR) \cite{Blaum2019,BR2019} codes were proposed
that are {\em $(n,k)$ recoverable} (i.e., any $k$ out of the $n=k+r$ columns can
retrieve all information symbols) and can locally repair any single-symbol failure
by accessing the symbols in a line of a certain slope, where $p$ is an odd prime.
Note that $(n,k)$ recoverable is reduced to MDS property if the number of information symbols
is $k$ times the number of symbols in a column.
%Each column in an EBR code has size $p$, where $p$ is an odd prime.
Generalized Expanded-Blaum-Roth (GEBR) codes \cite{wu2020,hou2021generalization} provide
a more general construction
such that the number of symbols in a column is not
required to be a prime number.

\subsection{Related Works}
EVENODD \cite{Blaum1995} and RDP \cite{Corbett2004Ro} are two systematic MDS array codes
correcting double column failures.
STAR codes \cite{4358290} extend the construction of EVENODD that can correct three column failures,
while generalized EVENODD \cite{blaum1996mds,blaum2002evenodd} is a more general construction of EVENODD with
three or more parity columns.
RDP are extended to be RTP \cite{RTP12} with three parity columns and generalized
RDP \cite{blaum2006family} with three or more parity columns.
Blaum-Roth (BR) codes \cite{Blaum1993} are constructed over the binary cyclic ring
$\mathbb{F}_2[x]/(1+x+\ldots+x^{p-1})$ by designing the parity-check
matrix to be a Vandermonde matrix that can support any parameters $k$ and $r$.
There are other studies that present new constructions of MDS array codes
\cite{746809,hou2014,Hou2018,7266845,hou2021star,feng2005} and
efficient decoding methods \cite{Hou2018form,hou2018a}.

In this paper, we propose Generalized Blaum-Roth (GBR) codes that
can support more parameters than that of BR codes, and show both the MDS condition
and the condition of recovering some erased lines
of any slope.

\subsection{Contributions}

In this paper, we study necessary and sufficient conditions for GEBR
codes to be $(n,k)$ recoverable and we also present sufficient conditions
for the recovery of some erased lines. Specifically,
our main contributions are as follows:
\begin{enumerate}
\item
%We give the MDS condition of GEBR codes for $\alpha\leq (p-1)\tau$.
When $\alpha\leq (p-1)\tau$, we show that $p\tau\times (k+r)$ GEBR codes %with $g(x)=1$
satisfy the $(n,k)$ recoverable property if $k+r\leq p^{\nu+1}$, where $\tau=\gamma p^{\nu}$ with
$\gcd(\gamma,p)=1$. Otherwise, if $k+r> p^{\nu+1}$ GEBR codes are not $(n,k)$ recoverable codes
under a specific condition (see Theorem \ref{thm:mds3} for details). Specifically, when $\gamma$ is an even number and
$p$ is a prime number with 2 being a primitive element in $\mathbb{Z}_p$,
we characterize the necessary and sufficient $(n,k)$ recoverable condition for GEBR codes
(see Theorem \ref{theo:g(x)} and Corollary \ref{coro:g(x)} for details).
\item
We present a sufficient condition for recovering some erased lines in
GEBR codes of any slope
$i$, $0\leq i\leq p\tau-1$, when $\tau$ is a power of $p$ (see Theorem \ref{thm:reco-tau1} for details).
The lines of slope $i$ are taken toroidally and an erased line of slope
$i$ means that the $k+r$ symbols in a line  of slope $i$ are erased.
\item
We propose GBR codes that extend the construction of BR codes to support more
parameters.
We show the necessary and sufficient MDS condition (see Theorem \ref{thm:mds4} for details)
and present a sufficient condition to recover some erased lines of any slope
(see Theorem \ref{thm:reco-tau2} for details).
\end{enumerate}

\subsection{Paper Organization}
The rest of the paper is organized as follows. In Section \ref{sec:cons}, we
review the construction of GEBR codes in \cite{hou2021generalization} and present the construction
of GBR codes. In Section \ref{sec:mds}, we show the $(n,k)$ recoverable condition for both GEBR codes
and GBR codes. In Section \ref{sec:lines}, we present two sufficient conditions of
recovering some erased lines for GEBR codes and GBR codes, respectively.
We conclude the paper in Section \ref{sec:con}.

\section{Constructions of GEBR Codes And GBR Codes}
\label{sec:cons}
In the section, we first review the construction of GEBR codes \cite{hou2021generalization}
and then present the construction of GBR codes.

\subsection{GEBR Codes}
A GEBR code \cite{hou2021generalization} consists of $m\times (k+r)$ arrays,
where each element in an array is in the finite field
$\mathbb{F}_{q}$, $q$ is a power of 2, $m=p\tau\geq k+r$ and $p$ is an
odd prime number. %, $\tau=\gamma p^{\nu}$ and $\gcd(\gamma,p)=1$.
It stores $\alpha\times k$ information symbols in the $k$
information columns with $\alpha$ information symbols each, for some
$\alpha\leq (p-1)\tau$, and uses the $\alpha\times k$ sub-array of information symbols as input
for encoding.

Given an $m\times (k+r)$ array, denote by
$s_{i,j}\in \mathbb{F}_{q}$ the element in
row $i$ and column $j$, where $i=0,1,\ldots,m-1$ and
$j=0,1,\ldots,k+r-1$. We assume that the
$\alpha k$ information symbols are $s_{i,j}$ with $i=0,1,\ldots,\alpha-1$ and
$j=0,1,\ldots,k-1$, while the remaining $m(k+r)-\alpha k$ symbols are
parity symbols.

For $j=0,1,\ldots,k+r-1$, we represent the $m$ symbols stored in column $j$
%(i.e., $s_{0,j},s_{1,j},\ldots,s_{m-1,j}$)
by a polynomial $s_j(x)$ of degree at most
$m -1$ over the ring $\mathbb{F}_{q}[x]$, i.e.,
\[
s_j(x)=s_{0,j}+s_{1,j}x+s_{2,j}x^2+\cdots+s_{m-1,j}x^{m-1},
\]
where $s_j(x)$ with $j=0,1,\ldots,k-1$ is an {\em information polynomial} and
$s_j(x)$ with $j=k,k+1,\ldots,k+r-1$ is a {\em parity polynomial}.  Let
$\mathcal{R}_{m}(q)=\mathbb{F}_{q}[x]/(1+x^{m})$ be the ring of polynomials
modulo $1+x^{m}$ with coefficients in $\mathbb{F}_{q}$.

Given the $\alpha$
information symbols $s_{0,j},s_{1,j},\ldots,s_{\alpha-1,j}$ in column
$j$, we compute
$m-\alpha$ symbols $s_{\alpha,j},s_{\alpha+1,j},\ldots,s_{m-1,j}$ for local
repair such that the polynomial $s_j(x)$ is a multiple of $(1+x^{\tau})g(x)$, where
$g(x)$ divides $1+x^{\tau}+\cdots+x^{(p-1)\tau}$.
%and $\gcd (g(x), 1+x^{\tau})=1$.
%We also require that $\gcd (g(x), h(x))=1$ and $\gcd (1+x^\tau, h(x))=1$.
Denoting by $\mathcal{C}_{p\tau}(g(x),\tau,q,d)$ the cyclic code of length $m=p\tau$ over
$\mathbb{F}_{q}$ with generator polynomial $(1+x^\tau)g(x)$ and minimum
distance $d$, then the information
polynomial $s_j(x)$ is in $\mathcal{C}_{p\tau}(g(x),\tau,q,d)$.
The coefficients of the polynomial
$s_{j}(x)=\sum_{\ell=0}^{m-1}s_{\ell,j}x^{\ell}$ in $\mathcal{C}_{p\tau}(g(x),\tau,q,d)$
must satisfy
\begin{equation}
\sum_{\ell=0}^{p-1}s_{\ell\tau+\mu,j}=0
\label{eq:tau-eqn}
\end{equation}
for all $\mu=0,1,\ldots,\tau-1$ \cite[Lemma 3]{hou2021generalization}.

The relationship between
the information polynomials and the parity polynomials is given by

\[
\mathbf{H}_{r\times (k+r)}\cdot\begin{bmatrix}
s_0(x) & s_1(x) & \cdots & s_{k+r-1}(x)
\end{bmatrix}^T=\mathbf{0}^T,
\]
where $\mathbf{H}_{r\times (k+r)}$ is the $r\times (k+r)$ parity-check matrix
\begin{equation}
\mathbf{H}_{r\times (k+r)}=
\begin{bmatrix}
 1&1&1 & \cdots & 1\\
 1 & x & x^{2}& \cdots & x^{k+r-1}\\
 \vdots& \vdots& \vdots & \ddots  & \vdots\\
 1& x^{r-1}& x^{2(r-1)}& \cdots & x^{(r-1)(k+r-1)}\\
 \end{bmatrix},
\label{eq:matrixH2}
\end{equation}
and $\mathbf{0}^T$ is an all-zero column of length $r$.  After solving the
above linear equations, the $r$ parity
polynomials are also in $\mathcal{C}_{p\tau}(g(x),\tau,q,d)$.  The
resulting codes with the parity-check matrix given
by Eq.~\eqref{eq:matrixH2} are denoted by $\textsf{GEBR}(p,\tau,k,r,q,g(x))$.

Let  $1+x^\tau+\cdots+x^{(p-1)\tau}=g(x)h(x)$. Then
the polynomial $h(x)$ is called the \emph{parity-check polynomial} of
$\mathcal{C}_{p\tau}(g(x),\tau,q,d)$, since the multiplication of any polynomial in
$\mathcal{C}_{p\tau}(g(x),\tau,q,d)$ and $h(x)$ is zero.
When $\gcd (g(x), h(x))=1$ and $\gcd (1+x^\tau, h(x))=1$,
the ring $\mathcal{R}_{p\tau}(q)$ is isomorphic to
$\mathbb{F}_{q}[x]/g(x)(1+x^\tau)\times \mathbb{F}_{q}[x]/(h(x))$
\cite[Theorem 1]{hou2021generalization}.

The ring $\mathcal{C}_{p\tau}(g(x),\tau,q,d)$ with $g(x)=1$ and $q=2$ has been used
to construct codes \cite{Hou2017,HOU2019,hou2019new} to have efficient repair.
The special ring with $\tau=1$ is employed to construct binary MDS
array codes \cite{Blaum1993,blaum1996mds,hou2014,hou2016,hou2018a,Hou2018form,BR2019}
that have lower computational complexity.

\subsection{GBR Codes}
A Generalized Blaum-Roth (GBR) code consists of $\alpha\times (k+r)$
arrays of symbols in $\mathbb{F}_{q}$, where $\alpha=\deg(h(x))=(p-1)\tau-\deg(g(x))$.
Recall that $h(x)$ is the parity-check polynomial of $\mathcal{C}_{p\tau}(g(x),\tau,q,d)$.
The first $k$ columns in an array store $\alpha k$ information symbols and the last
$r$ columns store $\alpha r$ parity symbols.
By representing the $\alpha$ symbols in each column of the $\alpha\times (k+r)$ array as a
polynomial in the ring $\mathbb{F}_{q}[x]/(h(x))$, we can obtain the
$r$ parity polynomials from the $k$ information polynomials
using the $r\times (k+r)$ parity-check matrix in Eq.
\eqref{eq:matrixH2}, where in this case
the operations are over the ring $\mathbb{F}_{q}[x]/(h(x))$.
We denote the array codes defined above by $\textsf{GBR}(p,\tau,k,r,q,g(x))$.
BR codes \cite{Blaum1993} are the special case of $\textsf{GBR}(p,\tau,k,r,q,g(x))$
codes in which $g(x)=1$ and $\tau=1$.

Both GEBR codes and GBR codes have the same parity-check matrix, the difference is that
GEBR codes operate over $\mathcal{C}_{p\tau}(g(x),\tau,q,d)$, while GBR codes
operate over $\mathbb{F}_{q}[x]/(h(x))$. When $\gcd(g(x),h(x))=\gcd(1+x^{\tau},h(x))=1$,
the ring $\mathcal{C}_{p\tau}(g(x),\tau,q,d)$
is isomorphic to $\mathbb{F}_{q}[x]/(h(x))$ \cite[Lemma 2]{hou2021generalization},
and we can obtain that the two codes share the same $(n,k)$ recoverable condition.
In GEBR codes, each column stores $p\tau$ symbols, i.e.,
there are $p\tau-\alpha$ local parity symbols in each column that can locally
repair up to $d-1$ symbol failures, while in GBR codes, each column stores $\alpha$ symbols
and there are no local parity symbols for local repair.
According to the parity-check matrix in Eq. \eqref{eq:matrixH2}, the summation of the $k+r$ symbols in
each line along slope $i$ in the $p\tau\times (k+r)$ array of GEBR codes is zero, where $0\leq i\leq r-1$.
If we add the all zero array of size $(p\tau-\alpha)\times (k+r)$ at the bottom of
an $\alpha\times (k+r)$ array in a GBR code to obtain a $p\tau\times (k+r)$ array,
then each line of slope $i$ of the obtained $p\tau\times (k+r)$ array has even parity
for $0\leq i\leq r-1$.

\begin{example}
Consider $g(x)=1$, $p=\tau=3$, $k=7$, $r=2$ and $\alpha=6$.
According to Eq. \eqref{eq:matrixH2}, the $2\times 9$ parity-check matrix of the example is
\[
\begin{bmatrix}
1 &1 &1 &1 &1 &1 &1 &1 &1 \\
1 &x &x^2 &x^3 &x^4 &x^5 &x^6 &x^7 &x^8 \\
\end{bmatrix}.
\]
Therefore, we have
\begin{align*}
s_7(x)+s_8(x)=&s_0(x)+s_1(x)+s_2(x)+s_3(x)+s_4(x)+s_5(x)+s_6(x),\\
x^7s_7(x)+x^8s_8(x)=&s_0(x)+xs_1(x)+x^2s_2(x)+x^3s_3(x)+x^4s_4(x)+x^5s_5(x)+x^6s_6(x).
\end{align*}
Suppose that the $k=7$ information polynomials are
\[
(s_0(x),s_1(x),s_2(x),s_3(x),s_4(x),s_5(x),s_6(x))=(1,x,x^2,x^3,x^4,x^5,1),
\]
we have
\begin{align*}
s_7(x)+s_8(x)=&x+x^2+x^3+x^4+x^5\bmod (1+x^3+x^6),\\
x^7s_7(x)+x^8s_8(x)=&1+x+x^4+x^5\bmod (1+x^3+x^6),
\end{align*}
and further obtain that
\[
(x+x^2+x^4+x^5)s_8(x)=x^3+x^4\bmod (1+x^3+x^6).
\]
Since the inverse of $x+x^2+x^4+x^5$ over $\mathbb{F}_2[x]/(1+x^3+x^6)$ is
$1+x+x^3+x^4+x^5$, we can compute that
\[
s_8(x)=(1+x+x^3+x^4+x^5)(x^3+x^4)=x^5\bmod (1+x^3+x^6),
\]
and further compute that
\[
s_7(x)=s_8(x)+\sum_{i=0}^{6}s_i(x)=x+x^2+x^3+x^4.
\]
Table \ref{table:example1} shows the example of $\textsf{GBR}(3,3,7,2,2,1)$.
We can check that each line of both slope zero and slope one of the $9\times 9$ array
in Table \ref{table:example1} has even parity.

\begin{table}[tbh]
\caption{Example of $\textsf{GBR}(3,3,7,2,2,1)$.}
\begin{center}
\begin{tabular}{|c|c|c|c|c|c|c|c|c|} \hline
% Column 0 &Column 1 & Column 2 & Column 3 & Column 4& Column 5& Column 6& Column 7& Column 8   \\ \hline \hline
 1  &0 &0 &0 & 0&0 &1 &0&0  \\ \hline
 0  &1 &0 &0 & 0&0 &0 &1&0\\ \hline
 0  &0 &1 &0 & 0&0 &0 &1&0\\ \hline
 0  &0 &0 &1 & 0&0 &0 &1&0\\ \hline
 0  &0 &0 &0 & 1&0 &0 &1&0\\ \hline
 0  &0 &0 &0 & 0&1 &0 &0&1\\ \hline \hline
 0  &0 &0 &0 & 0&0 &0 &0&0\\ \hline
 0  &0 &0 &0 & 0&0 &0 &0&0\\ \hline
 0  &0 &0 &0 & 0&0 &0 &0&0\\ \hline
\end{tabular}
\end{center}
\label{table:example1}
%\vspace{-0.5cm}
\end{table}
\end{example}

\section{The $(n,k)$ Recoverable Property}
\label{sec:mds}
When we say that an array code is $(n,k)$ recoverable or satisfies the $(n,k)$ recoverable property, it means that
the code can recover up to any $r$ erased columns out of the $k+r$ columns.
According to \cite[Lemma 7]{hou2021generalization},
$\textsf{GEBR}(p,\tau,k,r,q,g(x))$ codes are
$(n,k)$ recoverable if and only if the equation
\begin{equation}
(1+x^i)s(x)=0\bmod (1+x^{p\tau})
\label{eq:mds-cond}
\end{equation}
has a unique solution $s(x)=0\in\mathcal{C}_{p\tau}(g(x),\tau,q,d)$ for
$i\in \{1,2,\ldots,k+r-1\}$.

The following theorem states the $(n,k)$ recoverable
conditions of $\textsf{GEBR}(p,\tau,k,r,q,g(x))$ codes.

\begin{theorem}
Let $\tau=\gamma p^{\nu}$,
where $\nu\geq 0$, $0< \gamma$ and $\gcd (\gamma,p)=1$. Then,

\begin{description}

\item[(i)] If
$k+r\leq p^{\nu+1}$, then the codes
$\textsf{GEBR}(p,\tau,k,r,q,g(x))$ are $(n,k)$ recoverable.

\item[(ii)] If
$k+r>p^{\nu+1}$ and
$(1+x^{p^\nu}+(x^{p^\nu})^2+\cdots+(x^{p^\nu})^{p-1})\nmid g(x)$,
then the codes $\textsf{GEBR}(p,\tau,k,r,q,g(x))$ are not $(n,k)$ recoverable.

\end{description}

\label{thm:mds3}
\end{theorem}
\begin{proof}
Consider the first claim. In order to show that a
$\textsf{GEBR}(p,\tau,k,r,q,g(x))$ code is $(n,k)$ recoverable,
we have to show that Eq. \eqref{eq:mds-cond} has a unique solution $s(x)=0$ for all
$1\leq i\leq p^{\nu+1}-1$.

Let $i=up^s$, where $1\leq u\leq p^{\nu+1-s}-1$ and
$0\leq s\leq \nu$ such that $\gcd (u,p)=1$. Suppose that we can find a
non-zero polynomial $s(x)\in\mathcal{C}_{p\tau}(g(x),\tau,q,d)$ such that
Eq. \eqref{eq:mds-cond} holds.
%, we can deduce a contradiction as follows.
Without loss of generality, let $s(x)=\sum_{v=0}^{m-1}s_{v}x^{v}$
and $s_{0}=1$.  Then we have
\[
s_{((\ell-1)up^s)\bmod m}+s_{(\ell up^s)\bmod m}=0
\]
for $1\leq \ell\leq p^{\nu+1-s}-1$. By induction, we have $s_{(\ell up^s)\bmod m}=1$
for all $0\leq \ell\leq p^{\nu+1-s}-1$. Since $\gcd (u,p)=1$, the two sets
\begin{align*}
&\{(\ell up^s)\bmod m: 0\leq \ell\leq p^{\nu+1-s}-1\} \text{ and }\\
&\{(\ell p^s)\bmod m: 0\leq \ell\leq p^{\nu+1-s}-1\}
\end{align*}
coincide, and therefore we have $s_{\ell p^s}=1$ for $0\leq \ell \leq p^{\nu+1-s}-1$.
In particular, if $\ell=vp^{\nu-s}$ for $0\leq v\leq p-1$, then
\[
s_{\ell p^s}=s_{vp^{\nu-s} p^s}=s_{v\tau}=1.
\]
Therefore, $\sum_{v=0}^{p-1}s_{v\tau}=1$, contradicting $\sum_{v=0}^{p-1}s_{v\tau}=0$
(since $s(x)\in\mathcal{C}_{p\tau}(g(x),\tau,q,d)$).

Next we prove the second claim by showing that Eq.
\eqref{eq:mds-cond} has a non-zero solution $s(x)\in\mathcal{C}_{p\tau}(g(x),\tau,q,d)$
when $i=p^{\nu+1}$. Let
\begin{eqnarray}
&&s_0(x)=g(x)(1+x^{p^{\nu+1}}+x^{2p^{\nu+1}}+\cdots+x^{(\gamma-1)p^{\nu+1}}),\label{eq:mds-s0}\\
&&s_1(x)=g(x)(x^{p^{\nu}}+x^{p^{\nu}+p^{\nu+1}}+x^{p^{\nu}+2p^{\nu+1}}+\cdots+
x^{p^{\nu}+(\gamma-1)p^{\nu+1}})
\label{eq:mds-s1}
\end{eqnarray}
and $s(x)=s_0(x)+s_1(x)\bmod (1+x^m)$. We first show that the
polynomial $s(x)=s_0(x)+s_1(x)$ is in $\mathcal{C}_{p\tau}(g(x),\tau,q,d)$.

Since $s_1(x)=x^{p^\nu}s_0(x)$ and $1+x^m=1+x^{\gamma
p^{\nu+1}}$, we have $(1+x^{p^{\nu+1}})s_0(x)=g(x)(1+x^{\gamma
p^{\nu+1}})=0\bmod (1+x^m)$.
Since $1+x^\tau=1+x^{\gamma p^\nu}$ divides $1+x^{\gamma p^{\nu+1}}$
and $\gcd(\gamma, p)=1$, by the Euclidean algorithm,
$\gcd(1+x^{p^{\nu+1}}, 1+x^{\gamma p^\nu})=1+x^{p^\nu}$.
Since $(1+x^{p^{\nu+1}})(1+x^{p^{\nu+1}}+x^{2p^{\nu+1}}+\cdots+ x^{(\gamma-1)p^{\nu+1}})
=1+x^{\gamma p^{\nu+1}}=(1+x^{\gamma p^{\nu}})q(x)$ and
$\gcd(1+x^{p^{\nu+1}}, 1+x^{\gamma p^{\nu}})=1+x^{p^{\nu}}$, we have
$(1+x^{\gamma
p^{\nu}})|(1+x^{p^{\nu}})(1+x^{p^{\nu+1}}+x^{2p^{\nu+1}}+\cdots+
x^{(\gamma-1)p^{\nu+1}})$.
Hence,
$s(x)=s_0(x)+s_1(x)=g(x)(1+x^{p^{\nu}})(1+x^{p^{\nu+1}}+x^{2p^{\nu+1}}+\cdots+
x^{(\gamma-1)p^{\nu+1}})\bmod (1+x^m)$ is also divided by
$1+x^\tau=1+x^{\gamma p^{\nu}}$, i.e., $s(x)\in
\mathcal{C}_{p\tau}(g(x),\tau,q,d)$. It is clear that
$s_0(x)+s_1(x)\neq 0$. Hence, $s(x)=0\bmod (1+x^m)$ only when
$(1+x^{p^{\nu+1}})|g(x)(1+x^{p^{\nu}})$, i.e.,
$1+x^{p^{\nu}}+x^{2p^{\nu}}+\cdots+x^{(p-1)p^{\nu}}|g(x)$. However,
this is not true from the assumption in (ii), so it is not true such that $s(x)\neq 0$.

From the definitions of $s_0(x)$ and $s_1(x)$, we have
\begin{align*}
x^{p^{\nu+1}}s_0(x)=&g(x)(x^{p^{\nu+1}}+x^{2p^{\nu+1}}+\cdots+x^{\gamma p^{\nu+1}})\\
=&g(x)(x^{p^{\nu+1}}+x^{2p^{\nu+1}}+x^{3p^{\nu+1}}+\cdots+1)\\
=&s_0(x)\bmod (1+x^m),
\end{align*}
where the second equation above comes from $x^{\gamma
p^{\nu+1}}=x^{m}=1\bmod (1+x^m)$.
Similarly, we can obtain that $x^{p^{\nu+1}}s_1(x)=s_1(x)\mod
(1+x^m)$. Therefore, we have
\[
(1+x^{p^{\nu+1}})s(x)=(1+x^{p^{\nu+1}})(s_0(x)+s_1(x))=0\bmod (1+x^m),
\]
and $s(x)\neq 0$ is a solution to $(1+x^{p^{\nu+1}})s(x)=0\bmod (1+x^m)$,
proving the second claim.
\end{proof}

By Theorem \ref{thm:mds3}, the codes $\textsf{GEBR}(p,\tau,k,r,q,g(x))$ are $(n,k)$ recoverable if
$k+r\leq p^{\nu+1}$ and are not $(n,k)$ recoverable if $k+r>p^{\nu+1}$ and
$(1+x^{p^\nu}+(x^{p^\nu})^2+\cdots+(x^{p^\nu})^{p-1})\nmid g(x)$.
On the other hand, if $k+r>p^{\nu+1}$ and
$(1+x^{p^\nu}+(x^{p^\nu})^2+\cdots+(x^{p^\nu})^{p-1})\mid g(x)$, we
do not know whether the codes $\textsf{GEBR}(p,\tau,k,r,q,g(x))$ are $(n,k)$ recoverable.
Theorem 8 in \cite{hou2021generalization} is a special case of Theorem \ref{thm:mds3}
with $g(x)=1$.

When $g(x)$ is a power of $1+x^{p^\nu}+(x^{p^\nu})^2+\cdots+(x^{p^\nu})^{p-1}$ and $\gamma$
is an even number, the necessary and sufficient $(n,k)$ recoverable condition of the codes
$\textsf{GEBR}(p,\tau,k,r,q,g(x))$ is given in the next theorem.

\begin{theorem}
Let $\tau=\gamma p^{\nu}=2^j\ell p^{\nu}$ (i.e, $\gamma$ is an even number), where $j\geq 1$,
$\nu\geq 0$, $\gcd(\ell,2)\eq 1$ and $\gcd(\ell,p)\eq 1$. Let
%$g(x)=(1+x^{p^\nu}+(x^{p^\nu})^2+\cdots+(x^{p^\nu})^{p-1})^{2^i-1}$,
$g(x)=(1+x^{p^\nu}+(x^{p^\nu})^2+\cdots+(x^{p^\nu})^{p-1})^t$,
where $1\leq t\leq 2^j-1$.
Then,

\begin{description}

\item[(i)] If
$k+r\leq p^{\nu+1}$, then the codes
$\textsf{GEBR}(p,\tau,k,r,q,g(x))$ are $(n,k)$ recoverable.

\item[(ii)] If $k+r>p^{\nu+1}$, %and
%$(1+x^{p^\nu}+(x^{p^\nu})^2+\cdots+(x^{p^\nu})^{p-1})\mid g(x)$,
then the codes $\textsf{GEBR}(p,\tau,k,r,q,g(x))$ are not $(n,k)$ recoverable.

\end{description}

\label{theo:g(x)}
\end{theorem}
\begin{proof} It is enough to prove (ii), since (i) is a special case
of Theorem~\ref{thm:mds3}.
When $\tau=2^j\ell p^{\nu}$, we have $m=p\tau=2^j\ell p^{\nu+1}$ and
\begin{align*}
1+x^{\tau}+x^{2\tau}+\cdots+x^{(p-1)\tau}&=
1+x^{2^j\ell p^{\nu}}+x^{2^j\ell \cdot 2\cdot p^{\nu}}+\cdots+x^{2^j\ell (p-1)p^{\nu}}\\
&=(1+x^{\ell p^{\nu}}+x^{2\ell p^{\nu}}+\cdots+x^{(p-1)\ell p^{\nu}})^{2^j}.
\end{align*}
We claim that $(1+x^{p^{\nu}}+x^{2 p^{\nu}}+\cdots+x^{(p-1) p^{\nu}})|
(1+x^{\ell p^{\nu}}+x^{2\ell p^{\nu}}+\cdots+x^{(p-1)\ell p^{\nu}})$ as follows.
Recall that
\begin{align*}
1+x^{2^j\ell p^{\nu+1}}=&(1+x^{p^{\nu+1}})(1+x^{p^{\nu+1}}+x^{2p^{\nu+1}}+\cdots+x^{(2^j\ell-1)p^{\nu+1}})\\
=&(1+x^{2^j\ell p^{\nu}})(1+x^{2^j\ell p^{\nu}}+x^{2\cdot 2^j\ell p^{\nu}}+\cdots+x^{(p-1)2^j\ell p^{\nu}}).
\end{align*}
Since $\gcd(2^j\ell,p)=1$, we have $\gcd(1+x^{p^{\nu+1}},1+x^{2^j\ell p^{\nu}})=1+x^{p^{\nu}}$
and therefore
\[
1+x^{p^{\nu+1}}=(1+x^{p^{\nu}})(1+x^{p^{\nu}}+x^{2p^{\nu}}+\cdots+x^{(p-1)p^{\nu}})
\]
is a factor of $(1+x^{p^{\nu}})(1+x^{2^j\ell p^{\nu}}+x^{2\cdot 2^j\ell p^{\nu}}+\cdots+x^{(p-1)2^j\ell p^{\nu}})$, we can further obtain that $(1+x^{p^{\nu}}+x^{2 p^{\nu}}+\cdots+x^{(p-1) p^{\nu}})|
(1+x^{\ell p^{\nu}}+x^{2\ell p^{\nu}}+\cdots+x^{(p-1)\ell p^{\nu}})$.
Therefore, $g(x)\eq
(1+x^{p^\nu}+(x^{p^\nu})^2+\cdots+(x^{p^\nu})^{p-1})^{t}$ is a factor of
$1+x^{\tau}+x^{2\tau}+\cdots+x^{(p-1)\tau}$, where $1\leq t\leq 2^j-1$.

Let
\begin{eqnarray*}
s(x)&=&\left(\frac{1+x^{2^j\ell p^{\nu+1}}}{1+x^{p^{\nu+1}}}\right)\left(1+x^{p^{\nu}}\right),
\label{eqhs}
\end{eqnarray*}
which is in $\mathcal{R}_{p\tau}(q)$, since
\[
1+x^{2^j\ell p^{\nu+1}}=(1+x^{p^{\nu+1}})(1+x^{p^{\nu+1}}+x^{2\cdot p^{\nu+1}}+\cdots+x^{(2^j\ell-1)p^{\nu+1}}),
\]
is a multiple of $1+x^{p^{\nu+1}}$.
The degree of
\[
s(x)=\left(1+x^{p^{\nu+1}}+x^{2\cdot p^{\nu+1}}+\cdots+x^{(2^j\ell-1)p^{\nu+1}}\right)\left(1+x^{p^{\nu}}\right),
\]
is $(2^j\ell-1)p^{\nu+1}+p^{\nu}$ which is strictly less than $m=2^j\ell p^{\nu+1}$.
Therefore, $s(x)\bmod (1+x^{m})\neq 0$.

Since $(1+x^{p^{\nu+1}})s(x)\eq 0 \bmod (1+x^m)$, it remains to be proved
that $s(x)\in \mathcal{C}_{p\tau}(g(x),\tau,q,d)$, i.e., we have to
prove that $g(x)(1+x^{\tau})$ divides $s(x)$. Without loss of
generality, suppose that $t\eq 2^j-1$ and notice that

\begin{eqnarray*}
\frac{s(x)}{g(x)(1+x^{\tau})}&=&
\frac{\left(\frac{1+x^{2^j\ell
p^{\nu+1}}}{1+x^{p^{\nu+1}}}\right)\left(1+x^{p^{\nu}}\right)}
{\left( 1+x^{p^\nu}+(x^{p^\nu})^2+\cdots+(x^{p^\nu})^{p-1}\right)^{2^j-1}(1+x^{2^j\ell
p^{\nu}})}\\
&=&\frac{\left(\frac{\left(1+x^{\ell p^{\nu+1}}\right)^{2^j}}
{1+x^{p^{\nu+1}}}\right)\left(1+x^{p^{\nu}}\right)}
{\left(\frac{1+x^{p^{\nu +1}}}{1+x^{p^{\nu}}}\right)^{2^j-1}
\left(1+x^{\ell p^{\nu}}\right)^{2^j}}\\
%&=&\left(\frac{\left(1+x^{\ell
%p^{\nu+1}}\right)\left(1+x^{p^{\nu}}\right)}{\left(1+x^{p^{\nu+1}}\right)\left(1+x^{\ell
%p^{\nu}}\right)}\right)^{2^j}\\
%
%&=&\left(\frac{1+x^{\ell p^{\nu+1}}}
%{\frac{\left(1+x^{p^{\nu+1}}\right)\left(1+x^{\ell p^{\nu}}\right)}
%{1+x^{p^{\nu}}}\right)^{2^j},
&=&\left(\frac{1+x^{\ell p^{\nu+1}}}
{\frac{\left(1+x^{p^{\nu+1}}\right)\left(1+x^{\ell p^{\nu}}\right)}{1+x^{p^{\nu}}}}
\right)^{2^j},
\end{eqnarray*}
so, $g(x)(1+x^{\tau})$ divides $s(x)$ if and only if
%$\left(1+x^{p^{\nu+1}}\right)\left(\frac{1+x^{\ell p^{\nu}}}{1+x^{p^{\nu}}}\right)$
$\frac{\left(1+x^{p^{\nu+1}}\right)\left(1+x^{\ell p^{\nu}}\right)}{1+x^{p^{\nu}}}$
divides $1+x^{\ell p^{\nu+1}}$. In effect, notice that
\begin{eqnarray*}
\gcd(1+x^{p^{\nu+1}},1+x^{\ell p^{\nu}})\eq 1+x^{\gcd(p^{\nu+1}\,,\,\ell p^{\nu})}\eq 1+x^{p^{\nu}},
\end{eqnarray*}
so, since both
$1+x^{p^{\nu+1}}$ and $1+x^{\ell p^{\nu}}$ divide $1+x^{\ell p^{\nu+1}}$, then their
product divided by their largest common divisor also divides $1+x^{\ell p^{\nu+1}}$,
completing the proof.
\end{proof}

\begin{example}\label{exm:g(x)}
Consider the example of $p=5$, $j=1$, $\ell=1$, $\nu=1$ and $t=1$, we have $\tau=2^j \ell p^{\nu}=10$,
$m=2^j \ell p^{\nu+1}=50$ and $g(x)=1+x^5+x^{10}+x^{15}+x^{20}$.
By Theorem \ref{theo:g(x)}, the code $\textsf{GEBR}(5,10,k,r,q,1+x^5+x^{10}+x^{15}+x^{20})$
is $(n,k)$ recoverable if $k+r\leq 25$ and not $(n,k)$ recoverable if $k+r> 25$.

Assume that $k+r\leq 25$, the code is $(n,k)$ recoverable if and only if
$(1+x^i)s(x)=0\bmod (1+x^{50})$ has a unique solution $s(x)=0\in\mathcal{C}_{50}(1+x^5+x^{10}+x^{15}+x^{20},10,q,d)$ for
$i\in \{1,2,\ldots,24\}$. For example, take $i=20$, and suppose that
$s(x)=\sum_{\ell=0}^{49}s_{\ell}x^{\ell}\neq 0$. Without loss of generality, let
$s_0=1$, then we have $s_0=s_{20}=s_{40}=s_{10}=s_{30}=1$ and hence
$s_0+s_{20}+s_{40}+s_{10}+s_{30}=1$. This contradicts Eq. \eqref{eq:tau-eqn}.
Therefore, $(1+x^i)s(x)=0\bmod (1+x^{50})$ has a unique solution $s(x)=0$
in $\mathcal{C}_{50}(1+x^5+x^{10}+x^{15}+x^{20},10,q,d)$.

Suppose that $k+r> 25$, the code is not $(n,k)$ recoverable by Theorem \ref{theo:g(x)},
since we can find a non-zero solution $s(x)\in \mathcal{C}_{50}(1+x^5+x^{10}+x^{15}+x^{20},10,q,d)$
for $(1+x^{25})s(x)=0\bmod (1+x^{50})$.
Following the proof of Theorem \ref{theo:g(x)}, we can find that
$s(x)=(1+x^{25})(1+x^{5})\neq 0\bmod(1+x^{50})$.
It is easy to check that $(1+x^{25})s(x)=(1+x^{25})(1+x^{25})(1+x^{15})=0\bmod (1+x^{50})$
and $s(x)=(1+x^{25})(1+x^{15})\in \mathcal{C}_{50}(1+x^5+x^{10}+x^{15}+x^{20},10,q,d)$.
Therefore, the code is not $(n,k)$ recoverable when $k+r> 25$.
\end{example}

When 2 is a primitive element in $\mathbb{Z}_p$ and $q=2$, the polynomial
$1+x^{p^\nu}+(x^{p^\nu})^2+\cdots+(x^{p^\nu})^{p-1}$ is an irreducible polynomial
in $\mathbb{F}_2[x]$ \cite{Itoh1991Characterization}. Then $g(x)$ is a factor of
$(1+x^{p^\nu}+(x^{p^\nu})^2+\cdots+(x^{p^\nu})^{p-1})^{2^j}$ is reduced to that
$g(x)$ is a power of $1+x^{p^\nu}+(x^{p^\nu})^2+\cdots+(x^{p^\nu})^{p-1}$.
We can directly obtain the following results.

\begin{corollary}
Let $\tau=2^j\ell p^{\nu}$, where $j\geq 1$,
$\nu\geq 0$, $\gcd(\ell,2)\eq 1$ and $\gcd(\ell,p)\eq 1$.
Let $(1+x^{p^\nu}+(x^{p^\nu})^2+\cdots+(x^{p^\nu})^{p-1})|g(x)$.
Let $p$ be a prime number such that 2 is a primitive element in $\mathbb{Z}_p$ and $q=2$.
Then,

\begin{description}

\item[(i)] If
$k+r\leq p^{\nu+1}$, then the codes
$\textsf{GEBR}(p,\tau,k,r,q,g(x))$ are $(n,k)$ recoverable.

\item[(ii)] If $k+r>p^{\nu+1}$, %and
%$(1+x^{p^\nu}+(x^{p^\nu})^2+\cdots+(x^{p^\nu})^{p-1})\mid g(x)$,
then the codes $\textsf{GEBR}(p,\tau,k,r,q,g(x))$ are not $(n,k)$ recoverable.

\end{description}

\label{coro:g(x)}
\end{corollary}

Recall that codes $\textsf{GEBR}(p,\tau,k,r,q,g(x))$ and $\textsf{GBR}(p,\tau,k,r,q,g(x))$ share
the same parity-check matrix, the difference being that
$\textsf{GEBR}(p,\tau,k,r,q,g(x))$ is a code over
$\mathcal{C}_{p\tau}(g(x),\tau,q,d)$ and
$\textsf{GBR}(p,\tau,k,r,q,g(x))$ is a code over
$\mathbb{F}_{q}[x]/(h(x))$. Since the ring $\mathcal{C}_{p\tau}(g(x),\tau,q,d)$
is isomorphic to $\mathbb{F}_{q}[x]/(h(x))$ when
$\gcd(g(x),h(x))=\gcd(1+x^{\tau},h(x))=1$ \cite[Lemma 2]{hou2021generalization}, %the two codes
$\textsf{GEBR}(p,\tau,k,r,q,g(x))$ and $\textsf{GBR}(p,\tau,k,r,q,g(x))$ share
the same $(n,k)$ recoverable condition if $\gcd(g(x),h(x))=\gcd(1+x^{\tau},h(x))=1$. Therefore, the MDS condition of codes
$\textsf{GBR}(p,\tau,k,r,q,g(x))$  can be directly obtained from
Theorem~\ref{thm:mds3} as follows:
\begin{theorem}
Let $\tau=\gamma p^{\nu}$,
where $\nu\geq 0$, $0< \gamma$ and $\gcd (\gamma,p)=1$. If
$\gcd(g(x),h(x))=\gcd(1+x^{\tau},h(x))=1$, then,

\begin{description}

\item[(i)] If
$k+r\leq p^{\nu+1}$, then the codes
$\textsf{GBR}(p,\tau,k,r,q,g(x))$ are MDS.

\item[(ii)] If
$k+r>p^{\nu+1}$ and
$(1+x^{p^\nu}+(x^{p^\nu})^2+\cdots+(x^{p^\nu})^{p-1})\nmid g(x)$,
then the codes $\textsf{GBR}(p,\tau,k,r,q,g(x))$ are not MDS.

\end{description}
\label{thm:mds4}
\end{theorem}

When $g(x)=1$ and $\tau$ is a power of $p$, we have $h(x)=1+x^{\tau}+\cdots+x^{(p-1)\tau}$,
$\gcd(g(x),h(x))=1$ and $\gcd(1+x^{\tau},h(x))=1$ \cite[Theorem 10]{hou2021generalization}.
We can obtain that the necessary and sufficient MDS condition of codes
$\textsf{GBR}(p,\tau,k,r,q,g(x))$ with $k+r=p\tau$ is $\tau=p^{\nu}$, where $\nu\geq 0$.
When $\gcd(g(x),h(x))\neq 1$ or $\gcd(1+x^{\tau},h(x))\neq 1$, then
$\mathcal{C}_{p\tau}(g(x),\tau,q,d)$ is not isomorphic to $\mathbb{F}_{q}[x]/(h(x))$ and
we cannot obtain the MDS condition as in Theorem \ref{thm:mds4}.
The MDS condition of GBR codes with $\gcd(g(x),h(x))\neq 1$ or $\gcd(1+x^{\tau},h(x))\neq 1$
is a subject of future work.

\section{Recovery of Erased Lines of Slope $i$ in $\textsf{GEBR}(p,\tau,k,r,q,g(x))$ Codes}
\label{sec:lines}
In this section, we assume that $\textsf{GEBR}(p,\tau,k,r,q,g(x))$ codes
are $(n,k)$ recoverable. We will
show that $\textsf{GEBR}(p,\tau,k,r,q,1)$ codes can recover some erased
lines of slope $i$ under some constraint, where $i\in\{0,1,\ldots,p\tau-1\}$.
Similarly, we also present a sufficient condition of recovering some erased lines
of a slope for $\textsf{GBR}(p,\tau,k,r,q,g(x))$.
We first discuss some properties of the linear system over $\mathcal{R}_{p\tau}$
and then show the condition for recovering the erased lines.

\subsection{Linear System over $\mathcal{R}_{p\tau}$}
Let $\mathbf{V}_{t\times t}=[a_{i,j}(x)]$ be a $t\times t$ matrix with the entry
in row $i$ and column $j$ being $a_{i,j}(x)\in \mathcal{R}_{p\tau}$, where
$0\leq i,j\leq t-1$ and $t\geq 2$.
Consider the linear system
\begin{equation}
\mathbf{u}\mathbf{V}_{t\times t}=\mathbf{v}\bmod (1+x^{p\tau}),
\label{eq:linear-sys}
\end{equation}
where $\mathbf{u}=[u_0(x),u_1(x),\ldots,u_{t-1}(x)]\in\mathcal{R}_{p\tau}^t$ and
$\mathbf{v}=[v_0(x),v_1(x),\ldots,v_{t-1}(x)]\in\mathcal{R}_{p\tau}^t$.
If $a_{i,j}(x)$ is a power of $x$ and Eq. \eqref{eq:linear-sys} holds, one necessary requirement is that
\[
v_0(1)=v_1(1)=\cdots=v_{t-1}(1),
\]
as each polynomial $v_j(x)$ is computed by adding a cyclically shifted version
of $u_i(x)$.

\begin{theorem}
Assume that the determinant of $\mathbf{V}_{t\times t}$ and
$1+x^{\tau}+\cdots+x^{(p-1)\tau}$ are relatively prime.
Then, all the vectors $\mathbf{u}$ satisfying Eq. \eqref{eq:linear-sys} are
congruent to each other modulo $1+x^{\tau}+\cdots+x^{(p-1)\tau}$.
\label{thm:linear-sys-cong}
\end{theorem}
\begin{proof}
Since
\[
1+x^{\tau}+\ldots+x^{(p-1)\tau}=(1+x^{\tau})(x^{\tau}+x^{3\tau}+\ldots+x^{(p-2)\tau})+1,
\]
we obtain that $\gcd(1+x^{\tau},1+x^{\tau}+\ldots+x^{(p-1)\tau})=1$.
By the Chinese remainder theorem, we have an isomorphism
between $\mathcal{R}_{p\tau}(q)$ and $\mathbb{F}_{q}[x]/(1+x^\tau)\times
\mathbb{F}_{q}[x]/(1+x^{\tau}+\ldots+x^{(p-1)\tau})$.  The mapping $\theta$ is defined by
\[
\theta(b(x))=(b(x)\bmod (1+x^\tau), b(x) \bmod (1+x^{\tau}+\ldots+x^{(p-1)\tau})),
\]
where $b(x)\in \mathcal{R}_{p\tau}(q)$.
The inverse mapping $\theta^{-1}$ is
\begin{align*}
\theta^{-1}(b_1(x),b_2(x))=&\big( b_1(x)\cdot ((1+x^\tau+\ldots+x^{(p-1)\tau}))+\\
&b_2(x)\cdot (x^\tau+x^{2\tau}+\ldots+x^{(p-1)\tau}) \big) \bmod (1+x^{p\tau}),
\end{align*}
where $b_1(x)\in \mathbb{F}_{q}[x]/(1+x^\tau)$ and $b_2(x)\in \mathbb{F}_{q}[x]/(1+x^{\tau}+\ldots+x^{(p-1)\tau})$.
It suffices to investigate the following two linear systems
\begin{eqnarray}
\mathbf{u}\mathbf{V}_{t\times t}=&\mathbf{v}\bmod (1+x^{\tau}),\label{eq:linear-sys1}\\
\mathbf{u}\mathbf{V}_{t\times t}=&\mathbf{v}\bmod (1+x^{\tau}+\ldots+x^{(p-1)\tau}).\label{eq:linear-sys2}
\end{eqnarray}
In Eq. \eqref{eq:linear-sys1}, there are many solutions $\mathbf{u}^{'}$, each of which is
a polynomial over $\mathbb{F}_{q}[x]$ with degree less than $\tau$.

For Eq. \eqref{eq:linear-sys2}, since the determinant of $\mathbf{V}_{t\times t}$
is invertible modulo $1+x^{\tau}+\ldots+x^{(p-1)\tau}$ by the assumption, we can solve
Eq. \eqref{eq:linear-sys2} to obtain the unique solution $\mathbf{u}^{''}$.
Given the solutions $u^{'}_i(x)\in\mathbb{F}_{q}[x]/(1+x^\tau)$ and
$u^{''}_i(x)\in \mathbb{F}_{q}[x]/(1+x^{\tau}+\ldots+x^{(p-1)\tau})$ for Eq. \eqref{eq:linear-sys1}
and Eq. \eqref{eq:linear-sys2}, respectively, we can calculate one
solution for Eq. \eqref{eq:linear-sys} via the isomorphism $\theta^{-1}$ as
\begin{align*}
\theta^{-1}(u^{'}_i(x),u^{''}_i(x))=&\big( u^{'}_i(x)\cdot ((1+x^\tau+\ldots+x^{(p-1)\tau}))+\\
&u^{''}_i(x)\cdot (x^\tau+x^{2\tau}+\ldots+x^{(p-1)\tau}) \big) \bmod (1+x^{p\tau}),
\end{align*}
where $u^{'}_i(x)$ is a polynomial over $\mathbb{F}_{q}[x]/(1+x^{\tau})$
and $u^{''}_i(x)$ is the entry $i$ of the unique solution to Eq. \eqref{eq:linear-sys2}.
Therefore, there are many solutions to Eq. \eqref{eq:linear-sys2}, all of which are
unique after being reduced modulo $1+x^\tau+\ldots+x^{(p-1)\tau}$.
\end{proof}

By Theorem \ref{thm:linear-sys-cong}, even if the matrix $\mathbf{V}_{t\times t}$
is invertible over $\mathbb{F}_{q}[x]/(1+x^{\tau}+\ldots+x^{(p-1)\tau})$, we also have
many solutions $\mathbf{u}$ to Eq. \eqref{eq:linear-sys2}. However, if the solutions
$\mathbf{u}$ are reduced modulo $1+x^\tau+\ldots+x^{(p-1)\tau}$, then the results
are congruent. In other words, if any $\tau$ coefficients of each polynomial in
$\mathbf{u}$ are known, then the solution $\mathbf{u}$ to Eq. \eqref{eq:linear-sys2}
is unique. We summarize the results in the following lemma (without proof).
\begin{lemma}
Assume that $\tau$ coefficients of the polynomial $u_i(x)$ are known for $i=0,1,\ldots,t-1$
and the determinant of $\mathbf{V}_{t\times t}$ and $1+x^{\tau}+\cdots+x^{(p-1)\tau}$
are relatively prime. Then, Eq. \eqref{eq:linear-sys}
has a unique solution $\mathbf{u}$.
\label{lm:linear-sys-cong}
\end{lemma}

Theorem 1 in \cite{Hou2018form} is a special case of Theorem \ref{thm:linear-sys-cong}
with $\tau=1$ and $\mathbf{V}_{t\times t}$ a Vandermonde matrix.

According to Lemma \ref{lm:linear-sys-cong}, we can obtain the unique solution for
Eq. \eqref{eq:linear-sys} if we know any $\tau$ coefficients of the polynomial $u_i(x)$,
given that the condition in Lemma \ref{lm:linear-sys-cong} holds. In most of the cases,
we can obtain the unique solution if less than $\tau$ coefficients are known, since
we can calculate all the other coefficients of $u^{'}_i(x)$ when some coefficients are known.

\begin{example}
\label{exm:sys}
Consider an example of $p=\tau=3$ and $t=q=2$. Let $\mathbf{v}=[1+x+x^2+x^5, \text{ }1+x^3+x^7+x^8]$,
we want to solve $\mathbf{u}=[u_0(x),\text{ }u_1(x)]\in\mathcal{R}_{9}^2$ from the following linear system
\begin{equation}
\begin{bmatrix}
u_0(x)& u_1(x)
\end{bmatrix}\cdot\begin{bmatrix}
1& 1\\
1& x\\
\end{bmatrix}=
\begin{bmatrix}
1+x+x^2+x^5& 1+x^3+x^7+x^8
\end{bmatrix}\bmod (1+x^9).
\label{eq:exm}
\end{equation}
By Theorem \ref{thm:linear-sys-cong}, we can first solve
$\mathbf{u}^{'}=[u_0^{'}(x),\text{ }u_1^{'}(x)]\in\mathcal{R}_{3}^2$ from the following linear system
\begin{equation}
\begin{bmatrix}
u_0^{'}(x)& u_1^{'}(x)
\end{bmatrix}\cdot\begin{bmatrix}
1& 1\\
1& x\\
\end{bmatrix}=
\begin{bmatrix}
1+x+x^2+x^5& 1+x^3+x^7+x^8
\end{bmatrix}\bmod (1+x^3),
\label{eq:exm1}
\end{equation}
and solve
$\mathbf{u}^{''}=[u_0^{''}(x),\text{ }u_1^{''}(x)]\in\mathbb{F}^2_{2}[x]/(1+x^3+x^6)$ from the following linear system
\begin{equation}
\begin{bmatrix}
u_0^{''}(x)& u_1^{''}(x)
\end{bmatrix}\cdot\begin{bmatrix}
1& 1\\
1& x\\
\end{bmatrix}=
\begin{bmatrix}
1+x+x^2+x^5& 1+x^3+x^7+x^8
\end{bmatrix}\bmod (1+x^3+x^6).
\label{eq:exm2}
\end{equation}
For Eq. \eqref{eq:exm1}, we have two solutions, i.e., $\mathbf{u}^{'}=[0,\text{ }1+x]$
or $\mathbf{u}^{'}=[1+x+x^2,\text{ }x^2]$.
For Eq. \eqref{eq:exm2}, we have the unique solution $\mathbf{u}^{''}=[1+x+x^2+x^3+x^5,\text{ }x^3]$.
Then, we can obtain the two solutions in Eq. \eqref{eq:exm} as
$\mathbf{u}=[1+x^2+x^3+x^4+x^5+x^7,\text{ }x+x^3+x^4+x^7]$ or $\mathbf{u}=[x+x^6+x^8,\text{ }1+x^2+x^5+x^6+x^8]$.
Therefore, we only have two solutions in Eq. \eqref{eq:exm} and we can obtain the unique
solution if any coefficient of $u_0(x)$ or $u_1(x)$ is known.
\end{example}
\subsection{Recovery of Erased Lines for Any Slope}
Our idea is based on the observation that if we view the erased lines as polynomials over
$\mathcal{R}_{p\tau}$, then we can obtain a linear system over $\mathcal{R}_{p\tau}$
according to the parity-check matrix in Eq. \eqref{eq:matrixH2}.

\begin{example}
\label{exm:sploe}
Consider the $\textsf{GEBR}(p=3,\tau=3,k=4,r=2,q,g(x)=1)$ code. We have
$k+r=6$ polynomials $s_j(x)=\sum_{\ell=0}^{8}s_{\ell,j}x^{\ell}$ with $j=0,1,\ldots,5$
and the parity-check matrix is
\begin{equation}
\begin{bmatrix}
1 & 1 & 1 & 1 & 1 & 1\\
1 & x & x^2 & x^3 & x^4 & x^5\\
\end{bmatrix}.
\label{eq:parity-exm}
\end{equation}
From the first row of the matrix in Eq. \eqref{eq:parity-exm}, the summation of the
six symbols in each row (lines of slope zero) of the $9\times 6$ array is zero, i.e.,
\[
s_{\ell,0}+s_{\ell,1}+s_{\ell,2}+s_{\ell,3}+s_{\ell,4}+s_{\ell,5}=0,
\]
for $\ell=0,1,\ldots,8$. Similarly, by the second row of the matrix in Eq. \eqref{eq:parity-exm},
the summation of the symbols in each line of slope one is zero,
i.e.,
\[
s_{\ell,0}+s_{\ell-1,1}+s_{\ell-2,2}+s_{\ell-3,3}+s_{\ell-4,4}+s_{\ell-5,5}=0,
\]
for $\ell=0,1,\ldots,8$. Note that the indices are taken modulo $m=9$ in the example.
Suppose that four lines $e_1=0,e_2=1,e_3=3,e_4=4$ of slope $i=2$ are erased, i.e., the
following 24 symbols
\begin{align*}
&s_{0,0},s_{7,1},s_{5,2},s_{3,3},s_{1,4},s_{8,5},\\
&s_{1,0},s_{8,1},s_{6,2},s_{4,3},s_{2,4},s_{0,5},\\
&s_{3,0},s_{1,1},s_{8,2},s_{6,3},s_{4,4},s_{2,5},\\
&s_{4,0},s_{2,1},s_{0,2},s_{7,3},s_{5,4},s_{3,5},
\end{align*}
are erased.
For $\ell=0,1,\ldots,8$, we represent the six symbols $s_{\ell,0},s_{\ell-2,1},s_{\ell-4,2},
s_{\ell-4,3},s_{\ell-6,4},s_{\ell-8,5}$ of slope two by the polynomial
\[
\bar{s}_{\ell}(x)=s_{\ell,0}+s_{\ell-2,1}x+s_{\ell-4,2}x^2+
s_{\ell-6,3}x^3+s_{\ell-8,4}x^4+s_{\ell-10,5}x^5+s_{\ell-12,6}x^6+s_{\ell-14,7}x^7+s_{\ell-16,8}x^8,
\]
over $\mathbb{F}_{q}[x]/(1+x^{9})$, where $s_{\ell,6}=s_{\ell,7}=s_{\ell,8}=0$ for all
$\ell=0,1,\ldots,8$. If we can recover the four polynomials
$\bar{s}_{0}(x),\bar{s}_{1}(x),\bar{s}_{3}(x),\bar{s}_{4}(x)$, then the erased 24 symbols
are recovered.
By Eq. \eqref{eq:tau-eqn} with $p=\tau=3$ and $\mu=0,1$, we have
\begin{align*}
s_{0,j}+s_{3,j}=&s_{6,j},\\
s_{1,j}+s_{4,j}=&s_{7,j},
\end{align*}
and further obtain that
\begin{eqnarray}
\bar{s}_{0}(x)+\bar{s}_{3}(x)&=&\bar{s}_{6}(x),\label{eq:exm-col0}\\
\bar{s}_{1}(x)+\bar{s}_{4}(x)&=&\bar{s}_{7}(x).\label{eq:exm-col1}
\end{eqnarray}
According to the first row of Eq. \eqref{eq:parity-exm}, we have
\begin{align*}
0=&s_0(x)+s_1(x)+s_2(x)+s_3(x)+s_4(x)+s_5(x)\\
=&(s_{0,0}+s_{1,0}x+\cdots+s_{8,0}x^{8})+(s_{0,1}+s_{1,1}x+\cdots+s_{8,1}x^{8})+
\cdots+(s_{0,5}+s_{1,5}x+\cdots+s_{8,5}x^{8})\\
=&(s_{0,0}+s_{1,0}x+\cdots+s_{8,0}x^{8})+(s_{0,1}+s_{1,1}x+\cdots+s_{8,1}x^{8})+
\cdots+(s_{0,8}+s_{1,8}x+\cdots+s_{8,8}x^{8})\\
=&(s_{0,0}+s_{0,1}+\cdots+s_{0,8})+x(s_{1,0}+s_{1,1}+\cdots+s_{1,8})+
\cdots+x^{8}(s_{8,0}+s_{8,1}+\cdots+s_{8,8}),
\end{align*}
and therefore,
\begin{equation}
s_{\ell,0}+s_{\ell,1}+\cdots+s_{\ell,8}=0,
\label{eq:exa-slope0}
\end{equation}
for $\ell=0,1,\ldots,8$. Similarly, by the second row of Eq. \eqref{eq:parity-exm}, we have
\begin{align*}
0=&s_0(x)+xs_1(x)+x^2s_2(x)+\cdots+x^5s_5(x)\\
=&s_0(x)+xs_1(x)+\cdots+x^5s_5(x)+x^6s_6(x)+x^7s_7(x)+x^8s_8(x) \text{ since $s_6(x)=s_7(x)=s_8(x)=0$ }\\
=&(s_{0,0}+s_{1,0}x+\cdots+s_{8,0}x^{8})+x(s_{0,1}+s_{1,1}x+\cdots+s_{8,1}x^{8})+\\
&\cdots+x^8(s_{0,8}+s_{1,8}x+\cdots+s_{8,8}x^{8})\\
=&(s_{0,0}+s_{8,1}+\cdots+s_{1,8})+x(s_{1,0}+s_{0,1}+\cdots+s_{2,8})+\\
&\cdots+x^8(s_{8,0}+s_{7,1}+\cdots+s_{0,8}),
\end{align*}
and further obtain that
\begin{equation}
s_{\ell,0}+s_{\ell-1,1}+\cdots+s_{\ell-8,8}=0,
\label{eq:exa-slope1}
\end{equation}
for $\ell=0,1,\ldots,8$.
Then, we can compute that
\begin{eqnarray}
&&\bar{s}_{0}(x)+x^4\bar{s}_{1}(x)+x^{8}\bar{s}_{2}(x)+\cdots+
x^{32}\bar{s}_{8}(x)\nonumber\\
&=&(s_{0,0}+s_{7,1}x+s_{5,2}x^2+s_{3,3}x^3+s_{1,4}x^4+s_{8,5}x^5+s_{6,6}x^6+s_{4,7}x^7+s_{2,8}x^8)+\nonumber\\
&&x^4(s_{1,0}+s_{8,1}x+s_{6,2}x^2+s_{4,3}x^3+s_{2,4}x^4+s_{0,5}x^5+s_{7,6}x^6+s_{5,7}x^7+s_{3,8}x^8)+\cdots+\nonumber\\
&&x^{5}(s_{8,0}+s_{6,1}x+s_{4,2}x^2+s_{2,3}x^3+s_{0,4}x^4+s_{7,5}x^5+s_{5,6}x^6+s_{3,7}x^7+s_{1,8}x^8)\nonumber\\
&=&(s_{0,0}+s_{0,5}+\cdots+s_{0,4})+x(s_{7,1}+s_{7,6}+\cdots+s_{7,5})+\cdots+
x^{8}(s_{2,8}+s_{2,4}+\cdots+s_{2,3})\nonumber\\
&=&0,\label{eq:exm-row1}
\end{eqnarray}
where the last equation comes from Eq. \eqref{eq:exa-slope0}.
Similarly, we can compute that
\begin{eqnarray}
&&\bar{s}_{0}(x)+x^8\bar{s}_{1}(x)+x^{16}\bar{s}_{2}(x)+\cdots+
x^{64}\bar{s}_{8}(x)\nonumber\\
&=&(s_{0,0}+s_{7,1}x+s_{5,2}x^2+s_{3,3}x^3+s_{1,4}x^4+s_{8,5}x^5+s_{6,6}x^6+s_{4,7}x^7+s_{2,8}x^8)+\nonumber\\
&&x^8(s_{1,0}+s_{8,1}x+s_{6,2}x^2+s_{4,3}x^3+s_{2,4}x^4+s_{0,5}x^5+s_{7,6}x^6+s_{5,7}x^7+s_{3,8}x^8)+\cdots+\nonumber\\
&&x(s_{8,0}+s_{6,1}x+s_{4,2}x^2+s_{2,3}x^3+s_{0,4}x^4+s_{7,5}x^5+s_{5,6}x^6+s_{3,7}x^7+s_{1,8}x^8)\nonumber\\
&=&(s_{0,0}+s_{8,1}+\cdots+s_{1,8})+x(s_{7,1}+s_{6,2}+\cdots+s_{8,0})+\cdots+
x^{8}(s_{2,8}+s_{1,0}+\cdots+s_{3,7})\nonumber\\
&=&0,\label{eq:exm-row2}
\end{eqnarray}
where the last equation comes from Eq. \eqref{eq:exa-slope1}.
Combine with Eq. \eqref{eq:exm-col0}, Eq. \eqref{eq:exm-col1}, Eq. \eqref{eq:exm-row1}
and Eq. \eqref{eq:exm-row2}, we have
\begin{align*}
\begin{bmatrix}
\bar{s}_{e_1}(x) & \bar{s}_{e_2}(x) & \bar{s}_{e_3}(x) & \bar{s}_{e_4}(x)
\end{bmatrix}\begin{bmatrix}
1 & 0 & 1 & 1\\
0 & 1 & x^{4} & x^{8}\\
1 & 0 & x^{12} & x^{24}\\
0 & 1 & x^{16} & x^{32}\\
\end{bmatrix}=\begin{bmatrix}
\bar{s}_{6}(x) \\
\bar{s}_{7}(x) \\
\sum_{\ell\in\{2,5,6,7,8\}}x^{4\ell}\bar{s}_{\ell}(x) \\
\sum_{\ell\in\{2,5,6,7,8\}}x^{8\ell}\bar{s}_{\ell}(x) \\
\end{bmatrix}^T \bmod (1+x^9).
\end{align*}
Note that the determinant of the above $4\times 4$ matrix is $x+x^2+x^5+x^7$
after modulo $1+x^9$,
which is invertible over $\mathbb{F}_{q}[x]/(1+x^{3}+x^{6})$.
By Theorem \ref{thm:linear-sys-cong}, we can first solve the unique solution
$\bar{s}^{''}_{e_1}(x),\bar{s}^{''}_{e_2}(x),\bar{s}^{''}_{e_3}(x),\bar{s}^{''}_{e_4}(x)$
over $\mathbb{F}_{q}[x]/(1+x^{3}+x^{6})$ and then obtain
$\bar{s}_{i}(x)=\big( \bar{s}^{'}_i(x)\cdot ((1+x^3+x^{6}))+
\bar{s}^{''}_i(x)\cdot (x^3+x^{6}) \big) \bmod (1+x^{9})$,
where $i\in\{e_1,e_2,e_3,e_4\}$, $\bar{s}^{'}_i(x)\in\mathbb{F}_{q}[x]/(1+x^{3})$.
Note that the coefficients of degrees larger than five of the polynomial $\bar{s}_{i}(x)$
is zero, we can always find one unique polynomial $\bar{s}^{'}_i(x)\in\mathbb{F}_{q}[x]/(1+x^{3})$
for any polynomial $\bar{s}^{''}_i(x)\cdot (x^3+x^{6}) \bmod (1+x^{9})$ by Lemma \ref{lm:linear-sys-cong} such that
$\bar{s}_{6,i}=\bar{s}_{7,i}=\bar{s}_{8,i}=0$, where
$\bar{s}_{i}(x)=\sum_{\ell=0}^{8}\bar{s}_{\ell,i}$.
Specifically, let
\[
\sum_{\ell=0}^{8}a_{\ell,i}x^{\ell}=\big(\bar{s}^{''}_i(x)\cdot (x^3+x^{6}) \bmod (1+x^{9})\big),
\]
then $\bar{s}^{'}_i(x)=a_{6,i}+a_{7,i}x+a_{8,i}x^2$ and we have
\begin{align*}
\bar{s}_{i}(x)=&\bar{s}^{'}_i(x)\cdot ((1+x^3+x^{6}))+
\bar{s}^{''}_i(x)\cdot (x^3+x^{6}) \bmod (1+x^{9})\\
=&a_{6,i}+a_{7,i}x+a_{8,i}x^2+a_{6,i}x^3+a_{7,i}x^4+a_{8,i}x^5+a_{6,i}x^6+a_{7,i}x^7+a_{8,i}x^8
+\sum_{\ell=0}^{8}a_{\ell,i}x^{\ell}\\
=&\sum_{\ell=0}^{2}(a_{\ell,i}+a_{6+\ell,i})x^{\ell}+\sum_{\ell=3}^{5}(a_{\ell,i}+a_{3+\ell,i})x^{\ell},
\end{align*}
of which the coefficients of degrees larger than five are zero.
Therefore, we can recover the erased four lines.
\end{example}

Assume that $t$ lines $e_1,e_2,\ldots,e_t$ of slope $i$ are erased,
where $r\leq t\leq \tau+r$ and $0\leq i\leq p\tau-1$. We will present a sufficient condition for
recovering the $t$ erased lines.

In the $p\tau\times (k+r)$ array of the codes $\textsf{GEBR}(p,\tau,k,r,q,1)$, we have
$p\tau$ lines $0,1,\ldots,p\tau-1$ of slope $i$. We divide the $p\tau$ lines into
$\tau$ groups, for $g=0,1,\ldots,\tau-1$ such that group $g$ contains
$p$ lines $g,g+\tau,\ldots,g+(p-1)\tau$.
Let $\Theta$ be the subset of the $\tau$ groups
$\{0,1,\ldots,\tau-1\}$ containing the $t$ erased lines
and let $\eta=|\Theta|$ be the cardinality of $\Theta$.
For $g=0,1,\ldots,\eta-1$ and $h=1,2,\ldots,t$, we define $\mathbf{v}_g$ as the
row vector of length $t$ such that the entry $h$ is one if $e_h$ is in group
$\theta_g$ and zero otherwise, where $\Theta=\{\theta_0,\theta_1,\ldots,\theta_{\eta-1}\}$.
By representing the symbols in line $\ell$ with $\ell=0,1,\ldots,p\tau-1$ of slope $i$
by polynomial $\bar{s}_{\ell}(x)$, we can obtain the following polynomials
\[
\begin{bmatrix}
\mathbf{v}_0\\
\mathbf{v}_1\\
\vdots\\
\mathbf{v}_{\eta-1}\\
\end{bmatrix}
\begin{bmatrix}
\bar{s}_{e_1}(x)\\
\bar{s}_{e_2}(x)\\
\vdots\\
\bar{s}_{e_t}(x)\\
\end{bmatrix}
\]
according to Eq. \eqref{eq:tau-eqn}.
In Example \ref{exm:sploe}, the erased $t=4$ lines are located in $\eta=2$ groups $0,1$,
so we have $\Theta=\{0,1\}$, $\mathbf{v}_0=[1,0,1,0]$, $\mathbf{v}_1=[0,1,0,1]$,
obtaining
\[
\begin{bmatrix}
1& 0& 1& 0\\
0& 1& 0& 1\\
\end{bmatrix}
\begin{bmatrix}
\bar{s}_{0}(x)\\
\bar{s}_{1}(x)\\
\bar{s}_{3}(x)\\
\bar{s}_{4}(x)\\
\end{bmatrix}
\]
as in Eq. \eqref{eq:exm-col0} and Eq. \eqref{eq:exm-col1}.

Let $\mathcal{I}$ be the subset of $\{i,i-1,\ldots,i-(r-1)\}$ such that $i-\ell$ is a
multiple of $p$ for $\ell=0,1,\ldots,r-1$, where $i\in\{0,1,\ldots,p\tau-1\}$.
Let $\bar{\mathcal{I}}=\{i,i-1,\ldots,i-(r-1)\}\setminus
\mathcal{I}=\{g_1,g_2,\ldots,g_{|\bar{\mathcal{I}}|}\}$,
where $|\bar{\mathcal{I}}|$ is the cardinality of $\bar{\mathcal{I}}$. We have that
$|\bar{\mathcal{I}}|\geq r-\lceil \frac{r}{p}\rceil$.
In Example \ref{exm:sploe}, we have $i=2$ and $r=2$. Then, we have
$\mathcal{I}=\emptyset$ and $\bar{\mathcal{I}}=\{g_1,g_2\}=\{1,2\}$.
With the notation defined above, we can formulate some linear equations over
$\mathcal{R}_{p\tau}$ according to Eq. \eqref{eq:matrixH2}.

If $t\leq \tau$ and $\eta=t$, i.e, the erased $t$ lines are located
in $t$ groups, each having
one erased line. Then we can recover the erased $t$ lines one-by-one by Eq. \eqref{eq:tau-eqn}.
Otherwise, if $\eta<t$, it is possible to recover at most $\eta+r$ erased lines under some
condition, which will be shown in the next theorem.

\begin{theorem}
Let $\tau=p^{\nu}$ and $k+r\leq (p-1)\tau$, where $\nu\geq 0$.
The codes $\textsf{GEBR}(p,\tau,k,r,q,1)$
can recover $t$ erased lines $e_1,e_2,\ldots,e_t$ of slope $i$, if the following matrix
\begin{equation}
\mathbf{G}_{t\times t}=\begin{bmatrix}
 &  & \mathbf{v}_0 & \\
 &  & \mathbf{v}_1 & \\
 &  & \vdots & \\
 &  & \mathbf{v}_{\eta-1} & \\
x^{e_1(p\tau-g_1)^{-1}} & x^{e_2(p\tau-g_1)^{-1}} & \cdots & x^{e_t(p\tau-g_1)^{-1}}\\
x^{e_1(p\tau-g_2)^{-1}} & x^{e_2(p\tau-g_2)^{-1}} & \cdots & x^{e_t(p\tau-g_2)^{-1}}\\
\vdots & \vdots & \ddots & \vdots \\
x^{e_1(p\tau-g_{|\bar{\mathcal{I}}|})^{-1}} & x^{e_2(p\tau-g_{|\bar{\mathcal{I}}|})^{-1}} &
\cdots & x^{e_t(p\tau-g_{|\bar{\mathcal{I}}|})^{-1}}\\
\end{bmatrix}
\label{eq:thm-rec}
\end{equation}
is invertible over $\mathbb{F}_{q}[x]/(1+x^{\tau}+\ldots+x^{(p-1)\tau})$,
where $t=\eta+|\bar{\mathcal{I}}|$ and $\ell^{-1}\ell=1\bmod p\tau$
for integers $1\leq \ell,\ell^{-1}\leq p\tau-1$.
\label{thm:reco-tau1}
\end{theorem}
\begin{proof}
Assume that $t=\eta+|\bar{\mathcal{I}}|$ lines $e_1,e_2,\ldots,e_t$
of slope $i$ are erased, where $0\leq e_1 <e_2<\cdots <e_t\leq
p\tau-1$ and $0\leq i\leq p\tau-1$.
For $\ell=0,1,\ldots,p\tau-1$, we represent the $k+r$ symbols
$s_{\ell,0},s_{\ell-i,1},s_{\ell-2i,2},\ldots,s_{\ell-i(k+r-1),k+r-1}$
in line $\ell$ of slope $i$ by the polynomial
\[
\bar{s}_{\ell}(x)=s_{\ell,0}+s_{\ell-i,1}x+s_{\ell-2i,2}x^2+\ldots+s_{\ell-i(k+r-1),k+r-1}x^{k+r-1},
\]
which is in $\mathcal{R}_{p\tau}$\footnote{We can set
$s_{\ell-i(k+r),k+r}= s_{\ell-i(k+r+1),k+r+1}= \ldots =
s_{\ell-i(p\tau-1),p\tau-1}=0$.}.

According to the first row of the matrix in Eq. \eqref{eq:matrixH2},
we have that the summation of the $k+r$ symbols in the line of slope
zero is zero, i.e.,
\begin{equation}
s_{j,0}+s_{j,1}+s_{j,2}+\cdots+s_{j,k+r-1}=0
\label{eq:thm-rec-sum0}
\end{equation}
for $j=0,1,\ldots,p\tau-1$. Note that all the indices are taken modulo $p\tau$ in the proof
and $s_{j,k+r},s_{j,k+r+1},\cdots,s_{j,p-1}$ are all zero.
If $\gcd(i,p\tau)=\gcd(i,p^{\nu+1})=1$, then $\gcd(i,p)=\gcd(i,p\tau-i)=1$ and
\begin{eqnarray}
&&(p\tau-i)^{-1}\cdot (p\tau-i)=1\bmod p\tau \iff \nonumber\\
&&i\cdot(p\tau-i)^{-1}\cdot (p\tau-i)\cdot(p\tau-i)^{-1}=i\cdot(p\tau-i)^{-1}\bmod p\tau \iff \nonumber\\
&&i\cdot(p\tau-i)^{-1}=i\cdot(i)^{-1}\cdot(-1)^{-1}=(p\tau-1)\bmod p\tau \iff \nonumber\\
&&i\cdot(p\tau-i)^{-1}+1=0\bmod p\tau.\label{eq:slope-i}
\end{eqnarray}
We obtain that
\begin{align*}
&\bar{s}_{0}(x)+x^{(p\tau-i)^{-1}}\bar{s}_{1}(x)+\cdots+
x^{(p\tau-1)(p\tau-i)^{-1}}\bar{s}_{p\tau-1}(x)\\
=&\sum_{j=0}^{p\tau-1}s_{-ij,j}x^j+x^{(p\tau-i)^{-1}}(\sum_{j=0}^{p\tau-1}s_{1-ij,j}x^j)+
\cdots+x^{(p\tau-1)(p\tau-i)^{-1}}(\sum_{j=0}^{p\tau-1}s_{p\tau-1-ij,j}x^j)\\
=&(s_{0,0}+s_{-i,1}x+\cdots+s_{i,p\tau-1}x^{p\tau-1})+
x^{(p\tau-i)^{-1}}(s_{1,0}+s_{1-i,1}x+\cdots+s_{1+i,p\tau-1}x^{p\tau-1})\\
&+\cdots+x^{(p\tau-1)(p\tau-i)^{-1}}(s_{-1,0}+s_{-1-i,1}x+\cdots+s_{-1+i,p\tau-1}x^{p\tau-1})\\
=&(s_{0,0}+s_{1+i(p\tau-i)^{-1},1}+\cdots+s_{-1-i(p\tau-i)^{-1},p\tau-1})+\\
&x(s_{-i,1}+s_{1-i+i(p\tau-i)^{-1},1-(p\tau-1)^{-1}}+\cdots+s_{-1-i-i(p\tau-i)^{-1},1+(p\tau-i)^{-1}})\\
&+\cdots+x^{p\tau-1}(s_{i,p\tau-1}+s_{1+i+i(p\tau-i)^{-1},-(p\tau-1)^{-1}-1}+\cdots+s_{-1+i-i(p\tau-i)^{-1},-1+(p\tau-i)^{-1}})\\
=&(s_{0,0}+s_{0,1}+\cdots+s_{0,p\tau-1})+x(s_{-i,1}+s_{-i,1-(p\tau-1)^{-1}}+\cdots+s_{-i,1+(p\tau-i)^{-1}})\\
&+\cdots+x^{p\tau-1}(s_{i,p\tau-1}+s_{i,-(p\tau-1)^{-1}-1}+\cdots+s_{i,-1+(p\tau-i)^{-1}})\\
=&0,
\end{align*}
where the second to last equation above comes from Eq. \eqref{eq:slope-i}
and the last equation comes from Eq. \eqref{eq:thm-rec-sum0}.

Similarly, according to row $\ell$ of Eq. \eqref{eq:matrixH2},
the summation of the $k+r$ symbols in the line of slope $\ell$ is zero, where
$\ell=1,2,\ldots,r$.
If $\gcd(i-\ell,p\tau)=\gcd(i,p^{\nu+1})=1$, we have $\gcd(i-\ell,p)=\gcd(i-\ell,p\tau-i+\ell)=1$ and
obtain
\[
(i-\ell)\cdot(p\tau-i+\ell)^{-1}+1=0\bmod p\tau.
\]
We can further obtain that
\begin{equation}
\bar{s}_{0}(x)+x^{(p\tau-i+\ell)^{-1}}\bar{s}_{1}(x)+\cdots+
x^{(p\tau-1)(p\tau-i+\ell)^{-1}}\bar{s}_{p\tau-1}(x)=0.
\label{eq:rec-key}
\end{equation}
However, if $\gcd(i-\ell,p\tau)>1$ which means that $i-\ell$ is a multiple of $p$, then
$p\tau-i+\ell$ is a multiple of $p$, $(p\tau-i+\ell)^{-1}\bmod p\tau$ does not exist
and Eq. \eqref{eq:rec-key} does not hold.
Let $\{g_1,g_2,\ldots,g_{|\bar{\mathcal{I}}|}\}$ be the subset of
$\{i,i-1,\ldots,i-(r-1)\}$ such that $\gcd(i-\ell,p\tau)=1$ for $\ell=0,1,\ldots,r-1$,
which is defined before Theorem \ref{thm:reco-tau1}.
According to Eq. \eqref{eq:rec-key}, we can obtain the following polynomials
\begin{equation}
\begin{bmatrix}
x^{e_1(p\tau-g_1)^{-1}} & x^{e_2(p\tau-g_1)^{-1}} & \cdots & x^{e_t(p\tau-g_1)^{-1}}\\
x^{e_1(p\tau-g_2)^{-1}} & x^{e_2(p\tau-g_2)^{-1}} & \cdots & x^{e_t(p\tau-g_2)^{-1}}\\
\vdots & \vdots & \ddots & \vdots \\
x^{e_1(p\tau-g_{|\bar{\mathcal{I}}|})^{-1}} & x^{e_2(p\tau-g_{|\bar{\mathcal{I}}|})^{-1}} &
\cdots & x^{e_t(p\tau-g_{|\bar{\mathcal{I}}|})^{-1}}\\
\end{bmatrix}\cdot \begin{bmatrix}
\bar{s}_{e_1}(x)\\
\bar{s}_{e_2}(x)\\
\vdots\\
\bar{s}_{e_t}(x)\\
\end{bmatrix}.
\label{eq:rec-key1}
\end{equation}

Recalling that the summation of the $\tau$ symbols in rows $\ell,\ell+\tau,\ldots,\ell+(p-1)\tau$
in each column is zero according to Eq. \eqref{eq:tau-eqn},
we obtain the following polynomials
\begin{equation}
\begin{bmatrix}
\mathbf{v}_0\\
\mathbf{v}_1\\
\vdots\\
\mathbf{v}_{\eta-1}\\
\end{bmatrix}
\begin{bmatrix}
\bar{s}_{e_1}(x)\\
\bar{s}_{e_2}(x)\\
\vdots\\
\bar{s}_{e_t}(x)\\
\end{bmatrix}.
\label{eq:rec-row}
\end{equation}
Together with $|\bar{\mathcal{I}}|$ polynomials in Eq. \eqref{eq:rec-key1}, we have
the following $t=\eta+|\bar{\mathcal{I}}|$ polynomials
\[
\mathbf{G}_{t\times t}\cdot
\begin{bmatrix}
\bar{s}_{e_1}(x)\\
\bar{s}_{e_2}(x)\\
\vdots\\
\bar{s}_{e_t}(x)\\
\end{bmatrix},
\]
where $\mathbf{G}_{t\times t}$ is given in Eq. \eqref{eq:thm-rec}.

Recall that the last $\tau$ coefficients of $\bar{s}_{e_j}(x)$ are all zero which are known
for all $j=1,2,\ldots,t$.
If the $t\times t$ matrix $\mathbf{G}_{t\times t}$ is invertible in
$\mathbb{F}_q[x]/(1+x^{\tau}+\cdots+x^{(p-1)\tau})$, then we can compute the unique solution
$\bar{s}_{e_1}(x),\bar{s}_{e_2}(x),\ldots,\bar{s}_{e_t}(x)$ by Lemma \ref{lm:linear-sys-cong}.
Therefore, the condition for recovering the $t$ erased lines of slope $i$
is proved.
\end{proof}

If $\gcd(i-\ell,p\tau)=1$ for all $\ell=1,2,\ldots,r$, then we have
$|\bar{\mathcal{I}}|=r$. Since $\eta\leq \tau$, it is possible to recover
at most $\tau+r$ erased lines under some condition by Theorem \ref{thm:reco-tau1}.

Theorem 23-25 in \cite{hou2021generalization} is a special case
of Theorem \ref{thm:reco-tau1} with $r=2$ or some specific erased lines.

In codes $\textsf{GBR}(p,\tau,k,r,q,1)$, there is no local parity symbol in each column
and we can't obtain the polynomials in Eq. \eqref{eq:rec-row}. With a
proof similar to the one
of Theorem \ref{thm:reco-tau1}, we can obtain the recovery condition of
$\textsf{GBR}(p,\tau,k,r,q,1)$ as follows.

\begin{theorem}
Let $\tau=p^{\nu}$ and $k+r\leq (p-1)\tau$, where $\nu\geq 0$.
The codes $\textsf{GBR}(p,\tau,k,r,q,1)$
can recover $t$ erased lines $e_1,e_2,\ldots,e_t$ of slope $i$, if the following matrix
\[
\begin{bmatrix}
x^{e_1(p\tau-g_1)^{-1}} & x^{e_2(p\tau-g_1)^{-1}} & \cdots & x^{e_t(p\tau-g_1)^{-1}}\\
x^{e_1(p\tau-g_2)^{-1}} & x^{e_2(p\tau-g_2)^{-1}} & \cdots & x^{e_t(p\tau-g_2)^{-1}}\\
\vdots & \vdots & \ddots & \vdots \\
x^{e_1(p\tau-g_{|\bar{\mathcal{I}}|})^{-1}} & x^{e_2(p\tau-g_{|\bar{\mathcal{I}}|})^{-1}} &
\cdots & x^{e_t(p\tau-g_{|\bar{\mathcal{I}}|})^{-1}}\\
\end{bmatrix}
\]
is invertible over $\mathbb{F}_{q}[x]/(1+x^{\tau}+\ldots+x^{(p-1)\tau})$,
where $t=|\bar{\mathcal{I}}|$ and $\ell^{-1}\ell=1\bmod p\tau$
for integers $1\leq \ell,\ell^{-1}\leq p\tau-1$.
\label{thm:reco-tau2}
\end{theorem}

When $\tau=1$, $\textsf{GBR}(p,\tau=1,k,r,q,1)$ is reduced to BR codes.
If $i\in\{0,1,\ldots,r-1\}$, then $i-\ell$ is not a multiple of $p$ for
$\ell\in\{0,1,\ldots,r-1\}\setminus \{i\}$, we have that $\mathcal{I}=\{0\}$, $\bar{\mathcal{I}}=\{i,i-1,\ldots,1,-1,-2,\ldots, i-(r-1)\}$ and $|\bar{\mathcal{I}}|=r-1$.
We can obtain the following $r-1$ polynomials
\begin{align*}
\begin{bmatrix}
x^{e_1(p-g_1)^{-1}} & x^{e_2(p-g_1)^{-1}} & \cdots & x^{e_{r-1}(p-g_1)^{-1}}\\
x^{e_1(p-g_2)^{-1}} & x^{e_2(p-g_2)^{-1}} & \cdots & x^{e_{r-1}(p-g_2)^{-1}}\\
\vdots & \vdots & \ddots & \vdots \\
x^{e_1(p-g_{r-1})^{-1}} & x^{e_2(p-g_{r-1})^{-1}} &
\cdots & x^{e_t(p\tau-g_{r-1})^{-1}}\\
\end{bmatrix}\cdot \begin{bmatrix}
\bar{s}_{e_1}(x)\\
\bar{s}_{e_2}(x)\\
\vdots\\
\bar{s}_{e_{r-1}}(x)\\
\end{bmatrix}.
\end{align*}
Since $i\in\{0,1,\ldots,r-1\}$, according to row
$i$ of the parity-check matrix, the erased $r-1$ polynomials are in $\mathcal{C}_{p}(1,\tau=1,q,d)$
which is isomorphic to $\mathbb{F}_q[x]/(1+x+\ldots+x^{p-1})$.
Therefore, we can recover the erased $r-1$ lines of slope $i$ for $\textsf{GBR}(p,\tau=1,k,r,q,1)$
with $k+r=p$ if the determinant of the above $(r-1)\times(r-1)$ matrix is invertible in
$\mathbb{F}_q[x]/(1+x+\ldots+x^{p-1})$.

If $i\in\{r,r+1,\ldots,p-1\}$, we have that $\mathcal{I}=\emptyset$, $\bar{\mathcal{I}}=\{i,i-1,\ldots,i-(r-1)\}$ and $|\bar{\mathcal{I}}|=r$.
By Theorem \ref{thm:reco-tau2}, we can directly obtain that
BR codes with $k+r\leq p-1$ can recover $r$ erased lines $e_1,e_2,\ldots,e_t$ of slope $i$, if the following matrix
\[
\begin{bmatrix}
x^{e_1(p-i)^{-1}} & x^{e_2(p-i)^{-1}} & \cdots & x^{e_{r}(p-i)^{-1}}\\
x^{e_1(p-(i-1))^{-1}} & x^{e_2(p-(i-1))^{-1}} & \cdots & x^{e_{r}(p-(i-1))^{-1}}\\
\vdots & \vdots & \ddots & \vdots \\
x^{e_1(p-(i-r+1))^{-1}} & x^{e_2(p-(i-r+1))^{-1}} &
\cdots & x^{e_t(p-(i-r+1))^{-1}}\\
\end{bmatrix}
\]
is invertible over $\mathbb{F}_{q}[x]/(1+x+\ldots+x^{p-1})$.

When $\tau$ is not a power of $p$, the recovery condition of erased lines of slope $i$
is a subject of future research.

\section{Conclusions}
\label{sec:con}
In this paper, we presented the $(n,k)$ recoverable condition of GEBR codes for general $g(x)$.
We also showed a sufficient condition for recovering some erased lines of any slope in
GEBR codes when $\tau$ is a power of $p$.
Moreover, we presented the construction of GBR codes and showed the $(n,k)$ recoverable condition
and recovery condition for GBR codes.
The necessary and sufficient $(n,k)$ recoverable condition of GEBR codes for general
$g(x)$ will be a subject of future research, as well as
the recovery condition of erased lines of a slope with $\tau$ not being a power
of $p$.
%\appendices

\bibliographystyle{IEEEtran}
%\bibliography{references}

\end{document}